\documentclass[sn-mathphys,Numbered]{sn-jnl}

\usepackage{graphicx}%
\usepackage{multirow}%
\usepackage{amsmath,amssymb,amsfonts}%
\usepackage{amsthm}%
\usepackage{mathrsfs}%
\usepackage[title]{appendix}%
\usepackage{xcolor}%
\usepackage{textcomp}%
\usepackage{manyfoot}%
\usepackage{booktabs}%
\usepackage{algorithm}%
\usepackage{algorithmicx}%
\usepackage{epstopdf}%
\usepackage{algpseudocode}%
\usepackage{listings}%
\usepackage{tikz}%

\theoremstyle{thmstyleone}%
\theoremstyle{thmstyletwo}%
\newtheorem{lemma}{Lemma}%
\newtheorem{rk}{Remark}%
\theoremstyle{thmstylethree}%
\newtheorem{defn}{Definition}%
\newtheorem{thm}{Theorem}%
\newtheorem{pro}{Proposition}%

\begin{document}

\title[Alternative Gibbs measure for fertile three-state Hard-Core models on a Cayley tree]{Alternative Gibbs measure for fertile three-state Hard-Core models on a Cayley tree}


\author[1,2,3]{\fnm{Rustamjon} \sur{Khakimov}}\email{rustam-7102@rambler.ru}
\equalcont{These authors contributed equally to this work.}

\author*[1,2]{\fnm{Muhtorjon} \sur{Makhammadaliev}}\email{mmtmuxtor93@mail.ru}

\author[2]{\fnm{Kamola} \sur{Umrzakova}}\email{kamola-0983@mail.ru}
\equalcont{These authors contributed equally to this work.}

\affil[1]{\orgname{Institute of Mathematics, Academy of science}, \orgaddress{\street{University street 9}, \city{Tashkent}, \postcode{100174}, \country{Uzbekistan}}}

\affil[2]{\orgname{Namangan State University}, \orgaddress{\street{Boburshox street 161}, \city{Namangan}, \postcode{160107}, \country{Uzbekistan}}}

\affil[3]{\orgname{New Uzbekistan University}, \orgaddress{\street{Movarounnahr street 1}, \city{Tashkent}, \postcode{100000}, \country{Uzbekistan}}}


\abstract{We consider fertile three-state Hard-Core (HC) models with the activity parameter $\lambda>0$ on a Cayley tree. It is known that there exist four types of such models: "wrench"\,, "wand"\,, "hinge"\, and "pipe"\,. In cases "wand"\, and "hinge"\
on a Cayley tree of arbitrary order a complete description of translation-invariant Gibbs measures  is obtained. The conception of alternative Gibbs measure is introduced and in the case "wand"\, translational invariance conditions for alternative Gibbs measures are found. Also, we show that the existence of alternative Gibbs measures which are not translation-invariant.}




\maketitle

\section{Introduction}

The theory of Gibbs measures plays an important role in the study of the thermodynamic properties of physical and biological systems. The Gibbs measure is one of the main means of describing systems consisting of a large number of particles in statistical mechanics. It is known that each limiting Gibbs measure is associated with one phase of the physical system. Therefore, in the theory of Gibbs measures, one of the important problems is the existence of a phase transition, i.e. when a physical system changes its state with a change in temperature. This happens when the Gibbs measure is not unique. In this case, the temperature at which the state of a physical system changes is usually called critical. In addition, it is known that the set of all limit Gibbs measures forms a convex compact set (see, \cite {6}, \cite{Pr}, \cite{R}, \cite{Si}).

A Cayley tree $\Im^k$ of order $ k\geq 1 $ is an infinite tree, i.e.
a graph without cycles with exactly $k+1$ edges originate from each vertex. Let $\Im^k=(V,L,i)$, where $V$ is the set of vertices
$\Im^k$, $L$ is the set of its edges and $i$ is the incidence function,
which associates each edge $l\in L$ with its end points $x,y\in V$. If $i(l) = \{ x,y \} $, then $x$ and $y$ are called {\it nearest neighbors of the vertex} and are denoted by $l=\langle x,
y \rangle $.

For fixed $x^0\in V$ we set
$$W_n=\{x\in V\,| \, d(x,x^0)=n\}, \qquad V_n=\bigcup_{m=0}^n W_m,$$
where $d(x,y)$ is the distance between vertices $x$ and $y$ on the Cayley tree, i.e.
the number of edges of the shortest path connecting vertices $x$ and $y$.
We will write $x\prec y$ if the path from $x^0$ to $y$ goes through $x$.
A vertex $y$ is called a direct successor of vertex $x$ if $y\succ x$ and $x,y$
are the nearest neighbors. Note that in $\Im^k$ every vertex $x\neq x^0$
has $k$ direct successors, and the vertex $x^0$ has $k+1$ successors. We denote the successors of the vertex $x$ by $S(x)$, i.e. if $x\in W_n$, then
$$S(x)=\{y_i\in W_{n+1} | d(x,y_i)=1, i=1,2,\ldots, k \}.$$

\textbf{The \emph{HC}-model.}
Unlike other models, the Hard-Core (\emph{HC}) model places restrictions on spin values. Hard constraints arise in fields as diverse as combinatorics, statistical mechanics, and queuing theory. In particular, the \emph{HC}-model arises when studying random independent sets of the graph \cite {bw1}, \cite {gk}, when studying gas molecules on the lattice \cite{Ba}.
In queuing theory, a loss network is a stochastic model of communication networks in which calls are routed around a network between nodes. The links between nodes have finite capacity and thus some calls arriving may find no route available to their destination. These calls are lost from the network. The loss network was first studied by A.K. Erlang for a single telephone link (see \cite {ZZ}).

\emph{HC}-model arises as a simple example of a loss network with nearest-neighbor exclusion. The state $\sigma(x)$ at each node $x$ of the Cayley tree can be 0, 1 and 2. We have Poisson flow of calls of rate $\lambda$ at each site $x$, each call has an exponential
duration of mean 1. If a call finds the node in state 1 or 2 it is lost. If it finds the node in state 0 then things depend on the state of the neighboring sites. If all neighbors are in state 0, the call is accepted and the state of the node becomes 1 or 2 with equal probability $1/2$. If at least one neighbor is in state 1, and there is no neighbor in state 2 then the state of the node becomes 1. If at least one neighbor is in state 2 the call is lost.

Let $\Phi = \{0,1,2\}$ and $\sigma\in \Omega=\Phi^V$ be the configuration, i.e., $\sigma=\{\sigma(x)\in \Phi: x\in V\}$. In other words, in this model, each vertex
$x$ is assigned to one of the values $\sigma (x)\in \Phi=\{0,1,2\}$.
The values of $\sigma (x)\neq0$ mean that the vertex $x$ is `occupied', and
the value $\sigma (x)=0$ means that the vertex $x$ is `vacant'.
The set of all configurations on $V$ ($V_n$) is denoted by $\Omega$ ($\Omega_{V_n}$).

Consider the set $\Phi$ as the set of vertices of some graph
$G$. Using the graph $G$ we define a $G$-admissible configuration
in the following way. The $\sigma$ configuration is called
$G$-\textit{admissible configuration} on the Cayley tree (in $V_n$) if $\{\sigma (x),\sigma (y)\}$-edge of the graph $G$ for any
the nearest pair of neighbors $x,y$ from $V$ (from $V_n$). Let us denote the
set of $G$-admissible configurations via $\Omega^{a,G}$
($\Omega_{V_n}^{a,G}$).

The activity set \cite{bw} for the graph $G$ is the function $\lambda:G
\to R^3_+$. Value $\lambda_i$ of function $\lambda$ at vertex
$i\in\{0,1,2\}$ is called its ``activity''.

For given $G$ and $\lambda$, we define the Hamiltonian of the $G-$\emph{HC}-model as
 $$H^{\lambda}_{G}(\sigma)=\left\{%
\begin{array}{ll}
     \sum\limits_{x\in{V}}{\log \lambda_{\sigma(x)},} \ \ \ \mbox{if} \ \sigma \in\Omega^{a,G} \mbox{,} \\
   +\infty ,\ \ \ \ \ \ \ \ \ \  \ \ \ \mbox{if} \ \sigma \ \notin \Omega^{a,G} \mbox{.} \\
\end{array}%
\right. $$

\textbf{Finite-dimensional distributions.}
The concatenation of configurations $\sigma_{n-1}\in\Phi ^{V_{n-1}}$ and $\omega_n\in\Phi^{W_{n}}$ is defined by
$$
\sigma_{n-1}\vee\omega_n=\{\{\sigma_{n-1}(x), x\in V_{n-1}\},\{\omega_n(y), y\in W_n\}\}.
$$

Let $\mathbf{B}$ be the $\sigma$-algebra generated by
cylindrical sets with finite base of $\Omega^{a,G}$. For any $n$ we let $\mathbf{B}_{V_n}=\{\sigma\in\Omega^{a,G}:
\sigma|_{V_n}=\sigma_n\}$ denote the subalgebra of $\mathbf{B},$ where
$\sigma|_{V_n}-$  restriction of $\sigma$ to $V_n$ and $\sigma_n: x\in V_n
\mapsto \sigma_n(x)$ an admissible configuration in $V_n.$

\begin{defn}
For $\lambda >0$ the Gibbs \emph{HC}-measure is
probability measure $\mu$ on $(\Omega^{a,G} , \textbf{B})$ such that for
any $n\geq1$ and $\sigma_n\in \Omega^{a,G}_{V_n}$
$$
\mu \{\sigma \in \Omega^{a,G}:\sigma|_{V_n}=\sigma_n\}=
\int_{\Omega^{a,G}}\mu(d\omega)P_n(\sigma_n|\omega_{W_{n+1}}),
$$
where
$$
P_n(\sigma_n|\omega_{W_{n+1}})=\frac{e^{-H(\sigma_n)}}{Z_{n}
(\lambda ; \omega |_{W_{n+1}})}\textbf{1}(\sigma_n \vee \omega
|_{W_{n+1}}\in\Omega^{a,G}_{V_{n+1}}).
$$
\end{defn}
Here $Z_n(\lambda; \omega|_{W_{n+1}})$ is a normalization factor with boundary condition $\omega|_{W_n}$:
$$
Z_n (\lambda ; \omega|_{W_{n+1}})=\sum_{\widetilde{\sigma}_n \in \Omega_{V_n}} e^{-H(\widetilde{\sigma}_n)}\textbf{1}(\widetilde{\sigma}_n\vee
\omega|_{W_{n+1}}\in \Omega^{a,G}_{V_{n+1}}).
$$

In this paper we consider the case $\lambda_0=1, \ \lambda_1=\lambda_2=\lambda$.
Äëÿ $\sigma_n\in\Omega_{V_n}^{a,G}$ ïîëîæèì
$$\#\sigma_n=\sum\limits_{x\in V_n}{\mathbf 1}(\sigma_n(x)\geq 1),$$
i.e., the number of occupied vertices of $V_n$.

Let $z:\;x\mapsto z_x=(z_{0,x}, z_{1,x}, z_{2,x}) \in R^3_+$
vector-valued function on $V$. For $n\geq1$ and $\lambda>0$
consider the probability measure $\mu^{(n)}$ on $\Omega_{V_n}^{a,G}$,
defined as
\begin{equation}\label{e1}
\mu^{(n)}(\sigma_n)=\frac{1}{Z_n}\lambda^{\#\sigma_n} \prod_{x\in
W_n}z_{\sigma(x),x}.
\end{equation}

Here $Z_n$ is the normalizing divisor:
$$
Z_n=\sum_{{\widetilde\sigma}_n\in\Omega^{a,G}_{V_n}}
\lambda^{\#{\widetilde\sigma}_n}\prod_{x\in W_n}
z_{{\widetilde\sigma}(x),x}.
$$

The sequence of probability measures $\mu^{(n)}$ is said to be
consistent if for any $n\geq 1$ and
$\sigma_{n-1}\in\Omega^{a,G}_{V_{n-1}}$ the equality holds
\begin{equation}\label{e2}
\sum_{\omega_n\in\Omega_{W_n}}
\mu^{(n)}(\sigma_{n-1}\vee\omega_n){\mathbf 1}(\sigma_{n-1}\vee\omega_n\in\Omega^{a,G}_{V_n})=\mu^{(n-1)}(\sigma_{n-1}).
\end{equation}

\begin{defn}
 The measure $\mu$ that is the limit of a sequence $\mu^{(n)}$ defined by (\ref{e1}) with consistency condition (\ref{e2}) is called the splitting \emph{HC}-Gibbs measure (SGM) with $\lambda>0$ corresponding to the function $z:\,x\in V\setminus\{x^0\}\mapsto z_x$.
Moreover, an \emph{HC}-Gibbs measure corresponding to a constant function $z_x\equiv z$
is said to be translation-invariant (TI).
\end{defn}

\begin{defn}(\cite{bw})
The graph is called fertile if there is an activity set $\lambda$ such that the corresponding Hamiltonian has at least two TIGMs.
\end{defn}

\textbf{Problem statement.} The main task is to study the structure of the set $\mathcal G(H)$ of all Gibbs measures corresponding to a given Hamiltonian $H$.

A measure $ \mu \in \mathcal G(H) $ is called extreme if it cannot be expressed as $ \mu = \lambda \mu_1 + (1-\lambda)\mu_2$ for some $ \mu_1, \mu_2 \in \mathcal G(H) $ with $ \mu_1 \ne \mu_2 $.

As noted above, the set $ \mathcal G(H) $ of all Gibbs measures (for a given Hamiltonian $H$) is a nonempty convex compact set $ \mathcal G (H) $ in the space of all probability measures on $ \Omega$.

Using theorem (12.6) in \cite{6} and section 1.2.4 in \cite{BR} (see also Remark 2.7 in \cite{R2}), we can note the following.

\begin {itemize}
\item \ \ \emph{Any extreme Gibbs measure
$ \mu \in \mathcal G(H) $ is a SGM; therefore, the problem of describing Gibbs measures
reduces to describing the set of SGMs. For each fixed temperature, the description of the set $\mathcal G(H)$ is equivalent to a complete description of the set of all extreme SGMs, and hence we are only interested in SGMs on the Cayley tree.}

\item \ \ \emph{Any SGM corresponds to the solution of Eq. (\ref{e3}) (see below). Thus, our main
task reduces to solving functional equation (\ref{e3})}.
\end{itemize}

Let $L(G)$ be the set of edges of a graph $G$. We let $A\equiv A^G=\big(a_{ij}\big)_{i,j}$ denote the adjacency
matrix of the graph $G$, i.e.,

$$a_{ij}=a_{ij}^G=%
\begin{cases} 1, \ \ \mbox{if} \ \ \{i,j\}\in L(G), \\
0, \ \ \mbox{if} \ \ \{i,j\}\notin L(G).
\end{cases}
$$

The following theorem states a condition on $z_x$ that guarantees
consistency of the measure $\mu^{(n)}$.

\begin{thm}\label{thm1} \cite{Ro}
Probability measures
$\mu^{(n)}$, $n=1,2,\ldots$, given by the formula (\ref{e1}),
are consistent if and only if for any $x\in V$  the following equation holds:
\begin{equation}\label{e3}
\begin{cases}
z'_{1,x}=\lambda \prod_{y\in S(x)}{a_{10}+a_{11}z'_{1,y}+a_{12}z'_{2,y}\over a_{00}+a_{01}z'_{1,y}+a_{02}z'_{2,y}},\\
z'_{2,x}=\lambda \prod_{y\in S(x)}{a_{20}+ a_{21}z'_{1,y}+a_{22}z'_{2,y}\over a_{00}+a_{01}z'_{1,y}+a_{02}z'_{2,y}},
\end{cases}
\end{equation}
where $z'_{i,x}=\lambda z_{i,x}/z_{0,x}, \ \ i=1,2$.
\end{thm}

Let $G_k$ be a free product of $k+1$ cyclic groups $\{e,a_i\}$ of order two with the respective generators $a_1,a_2,...,a_{k+1}, a_i^2=e$.
There is a one-to-one correspondence between the set of vertices $V$ of the Cayley
tree of order $k$ and the group $G_k$ (see. \cite{11}, \cite{GR2003}, \cite{1}).

Let $\widehat{G}_k$ be a normal divisor of a finite index $r\geq 1$ and $G_k/\widehat{G}_k=\{H_1,...,H_r\}$ be the quotient group.

\begin{defn}
A collection of quantities $z=\{z_x,x\in G_k\}$ is said to be $\widehat{G}_k$-periodic if $z_{yx}=z_x$ for $\forall x\in G_k, y\in\widehat{G}_k.$ The $G_k$-periodic collections are called translation invariant.
\end{defn}

%

\begin{defn}
 A measure $\mu$ is called $\widehat{G}_k$-periodic if it  corresponds to a $\widehat{G}_k$-periodic collection of quantities $z$.
\end{defn}

\textbf{History of the study of SGMs for the \emph{HC}-model.} We present a brief overview of the work related to the \emph{HC}-model on the Cayley tree.

In \cite{Maz} A. Mazel and Yu. Suhov introduced and studied the \emph{HC}-model on the $d$-dimensional lattice $\mathbb Z^d$. In \cite{SR} this model with two states is considered  on the Cayley tree (see also the references in \cite{RKhM} for more details).
In the present paper, we study this model with three states on the Cayley tree.

In \cite{bw}, fertile three-state \emph{HC}-models are identified that correspond to graphs of the hinge, pipe, wand, and wrench types. The works \cite{XR1}-\cite{MRS}, \cite{Ro} are devoted to the study of Gibbs measures for \emph{HC}-models with three states on the Cayley tree.
In particular, in cases $G=\textit{wand}$ and $G=\textit{hinge}$ the following facts are known :
\begin{itemize}
\item[$\bullet$] In cases $G=\textit{wand}$ and $G=\textit{hinge}$ a complete description of the TISGM is obtained on the Cayley tree of order two (resp. three) (see \cite{Ro}, resp. \cite{XR1}).

\item[$\bullet$] On the Cayley tree of order $k>3$ for $\lambda\leq\lambda_{cr}$, there is a unique TISGM and for $\lambda>\lambda_{cr}$ there are at least three TISGMs where $\lambda_{cr}={1\over k-1}\cdot\left({2\over k}\right)^k$ in the case $G=\textit{wand}$ and  $\lambda_{cr}={1\over k-1}\cdot\left({k+1\over k}\right)^k$ in the case $G=\textit{hinge}$ (see \cite{RKh1}).

\item[$\bullet$] In cases $G=\textit{wand}$ and $G=\textit{hinge}$ for $k=2$ the regions of extremes and non-extremes of TISGMs are found (see \cite{RKh1}).

\item[$\bullet$] In cases $G=\textit{wand}$ and $G=\textit{hinge}$ for $k=3$ the regions where the unique TISGM is (not) extremal and the conditions under which the extremal measure is not unique are found (see \cite{KhU}).
%

\item[$\bullet$] In the case $G=\textit{wand}$ for $k\geq2$ conditions under which there are $G^{(2)}_k$-periodic splitting Gibbs measures (non translation-invariant) are found (see \cite{KhU1}).

\item[$\bullet$] In the case $G=\textit{wand}$ for $k=2$ the regions of of extremes and non-extremes of $G^{(2)}_k$-periodic splitting Gibbs measures are found (see \cite{KhU1}).

\item[$\bullet$] In the case $G=\textit{hinge}$ the non-uniqueness of periodic Gibbs measures under certain conditions is proved on a Cayley tree of order $k\geq6$ (see \cite{Ro}).

\item[$\bullet$] In the case $G=\textit{hinge}$ for $k\geq2$ translation-invariance conditions of $G^{(2)}_k$-periodic splitting Gibbs measure are found (see \cite{KhU2}).

\end{itemize}

In this paper, we study the \emph{HC}-model with three states on the Cayley tree in the cases of fertile graphs of the "wand"\, and "hinge"\, types. In these cases, the Conjectures from \cite{RKh1} are proven, i.e. a complete description of the TIGM is obtained for a HC model with three states on a Cayley tree of arbitrary order. In addition, the concept of an alternative Gibbs measure is introduced. Translational invariance conditions for alternative Gibbs measures are found. Moreover, the existence of alternative Gibbs measures that are not translation invariant is proven.

\section{Translation-invariant Gibbs measures}\label{sec2}

In (\ref{e3}) we assume that $z_{0,x} \equiv1$ è $z_{i,x}=z'_{i,x}>0, \  \  i=1,2.$ Then by Theorem 1 there is a unique $G$-\emph{HC}-Gibbs measure $\mu$ if and only if the equality holds for any function $z:x\in V \longmapsto z_{x}=(z_{1,x},z_{2,x})$:
\begin{equation}\label{e4}
z_{i,x}=\lambda\prod_{y\in S(x)}\frac{a_{i0}+a_{i1}z_{1,y}+a_{i2}z_{2,y}}{a_{00}+a_{01}z_{1,y}+a_{02}z_{2,y}},   \ \  i=1,2.
\end{equation}

We consider fertile graphs $G=\textit{wand}$ and $G=\textit{hinge}$ (see Fig. 1.):
$$
\begin{array}{ll}
\mbox{\it wand}: &  \{0,1\}\{0,2\}\{1,1\}\{2,2\}\\
\mbox{\it hinge}: &  \{0,0\}\{0,1\}\{0,2\}\{1,1\}\{2,2\}.\\
\end{array} $$

\begin{figure}[h]\label{fig1}
\center{\includegraphics[scale=0.5]{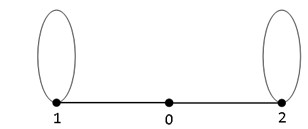} \qquad \includegraphics[scale=0.5]{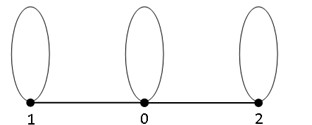}}
\caption{Fertile graphs $G=\textit{wand}$ (on the left) and $G=\textit{hinge}$ (on the right).}
\end{figure}

In \cite{RKh1}, the conjecture that in the case $G=\textit{wand}$ for any $k\geq4$ and $\lambda>\lambda_{cr}$ there are only three TISGMs is stated (see \cite{RKh1}, Conjecture 2). We will prove this conjecture.

We consider translation invariant solutions with $z_x=z\in R^2_+$, $x\neq x_0$.

In \cite{Ro}, for TISGM in the case of $G=\textit{wand}$ the following system of equations is considered:
\begin{equation}\label{e5}
\begin{cases}
    z_{1}=\lambda\left(\frac{1+z_{1}}{z_{1}+z_{2}}\right)^k, \\
    z_{2}=\lambda\left(\frac{1+z_{2}}{z_{1}+z_{2}}\right)^k.
\end{cases}
\end{equation}

In the case $z_{1}=z_{2}=z$ from the system of equations (\ref{e5}) we obtain
\begin{equation}\label{e6}
z=\lambda\Big(\frac{1+z}{2z}\Big)^k.
\end{equation}
The following lemma is known.

\begin{lemma} \cite{Ro} 
 Let $k\geq2$. Then for any $\lambda>0$ the equation (\ref{e10}) has a unique positive solution.
\end{lemma}

Now we prove Conjecture 2 given in \cite{RKh1}, i.e. the following theorem is true.

\begin{thm}\label{thm2}
Let $k\geq2$ and $\lambda_{cr}^{(1)}(k)=\frac{2^{k}}{(k-1)k^k}$. Then for the \emph{HC}-model in the case of $G=\textit{wand}$ for $0<\lambda\leq\lambda_{cr}^{(1)}$ there is exactly one  TISGM $\mu_0$, and for $\lambda>\lambda_{cr}^{(1)}$ there are exactly three TISGMs $\mu_0$, $\mu_1$ and $\mu_2$.
\end{thm}

\begin{proof}
We prove that there exists $\lambda_{cr}^{(1)}$ such that the system of equations (\ref{e5}) for $\lambda\leq\lambda_{cr}^{(1)}$ has a unique positive solution, and for $\lambda>\lambda_{cr}^{(1)}$ has exactly three solutions and we find the explicit form of $\lambda_{cr}^{(1)}$. For convenience, we denote $z_{1}=x$ and $z_{2}=y$.

From (\ref{e5}) we get
\begin{equation}\label{e7}
\frac{x}{y}=\left(\frac{1+x}{1+y}\right)^k.
\end{equation}
We introduce the notation $\frac{1+x}{1+y}=t$, $(t>0)$.
Then by virtue of (\ref{e7}) ($x=y\cdot t^k$), after some algebras we obtain
$$\left(t-1\right)\left(y\cdot(t^{k-1}+t^{k-2}+\dots+t) -1\right)=0.$$
Hence $t=1$ or
$$y\cdot(t^{k-1}+t^{k-2}+\dots+t) -1=0.$$
It is clear that for $t=1$ we have the solution $x=y=x^*$. By Lemma 1 it follows that for any $\lambda>0$
 a solution of this kind is unique. Then
$$y(t)=\frac{1}{t^{k-1}+t^{k-2}+\dots+t}$$
and the function $y(t)$ is uniquely determined for each value of $t$, since $y'(t)<0.$
The values of $x(t)$ corresponding to each value of $y(t)$ are determined by the formula
$$x(t)= t^k\cdot y(t)=\frac{t^k}{t^{k-1}+t^{k-2}+\dots+t}.$$
We substitute the expressions for $x(t)=y(t)\cdot t^k$ and $y(t)$ into the first equation $(\ref{e5})$. Then
\begin{equation}\label{e8}
\frac{t^k}{t^{k-1}+t^{k-2}+\dots+t}=\lambda\left(t^k+t^{k-1}+t^{k-2}+\dots+t \over t^k+1\right)^k.
\end{equation}
Solving the equation $(\ref{e8})$ with respect to the variable $t=t(\lambda,k)$ seems to be difficult. Therefore, we consider the equation (\ref{e8}) with respect to the variable $\lambda$. From (\ref{e8}) we find
\begin{equation}\label{e9}
\lambda(t)=\frac{(t^k+1)^k}{(t^{k-1}+t^{k-2}+\dots+t)(t^{k-1}+t^{k-2}+\dots+1)^k}=
\frac{(t^k+1)^k}{\Big(\sum_{i=1}^{k-1}t^i\Big)\Big(\sum_{i=0}^{k-1}t^i\Big)^k}.
\end{equation}
Let us prove that each value $\lambda$ corresponds to only one value $t$. Note that if $t$ is a solution of (\ref{e8}), then $\frac{1}{t}$ is also a solution to (\ref{e8}). Hence, it suffices to show that each value of $\lambda$ corresponds to exactly one value $t>1$ (or $t<1$).
To do this, consider the derivative of the function $\lambda(t)$:
$$\lambda'(t)=\frac{(t^k+1)^{k-1}\cdot\mathcal{R}(t,k)}{t^2(t^k-1)(t^{k-1}-1)^2\cdot(t^{k-1}+t^{k-2}+\dots+t+1)^k},$$
where
$$\mathcal{R}(t,k)=(t^k-1)(t^{2k}-1)+kt(t^{k-2}-1)(t^{2k}-1)-2k^2t^k(t-1)(t^{k-1}-1).$$
It can be seen that $t=1$ is the double root of the polynomial $\mathcal{R}(t,k)$. We show that $t=1$ is four-fold root of $\mathcal{R}(t,k)$. To do this, we introduce the notation $t-1=x$ ($t=x+1$) and prove that $ x=0$ is four-fold root of $\mathcal{R}_1(x,k)=\mathcal{R}_1(t-1,k)$, i.e. in the polynomial $\mathcal{R}_1(x,k)$ the smallest degree of the variable $x$ is four. We rewrite $\mathcal{R}_1(x,k)$ as follows:
$$\mathcal{R}_1(x,k)=\left(\sum_{i=1}^{2k}{C_{2k}^ix^i}\right)\left(\sum_{i=1}^{k}{C_{k}^ix^i}+k(x+1)\sum_{i=1}^{k-2}{C_{k-2}^ix^i}\right)
-2k^2\left(\sum_{i=0}^{k}{C_{k}^ix^{i+1}}\right)\left(\sum_{i=1}^{k-1}{C_{k-1}^ix^i}\right).$$
It is easy to show that the coefficients at $x^2$ and $x^3$ in the polynomial $\mathcal{R}_1(x,k)$ are equal to zero, and at $x^4$ they are non-zero.

Now we show that $\mathcal{R}_1(x,k)\geq0$, i.e.
\begin{equation}\label{e10}
\left(\sum_{i=1}^{2k}{C_{2k}^ix^i}\right)\left(\sum_{i=1}^{k}{C_{k}^ix^i}+k(x+1)\sum_{i=1}^{k-2}{C_{k-2}^ix^i}\right)
>\left(2k\sum_{i=0}^{k}{C_{k}^ix^{i+1}}\right)\left(k\sum_{i=1}^{k-1}{C_{k-1}^ix^i}\right).
\end{equation}
Firstly, we show that in (\ref{e10}) the expression in the second factor of the LHS is greater than the expression in the second factor of the  RHS, i.e. validity of the inequality
$$\sum_{i=1}^{k}{C_{k}^ix^i}+k(x+1)\sum_{i=1}^{k-2}{C_{k-2}^ix^i}>k\sum_{i=1}^{k-1}{C_{k-1}^ix^i}.$$
Indeed,
$$\sum_{i=1}^{k}{C_{k}^ix^i}+k(x+1)\sum_{i=1}^{k-2}{C_{k-2}^ix^i}=\sum_{i=1}^{k}{C_{k}^ix^i}+k\sum_{i=1}^{k-2}{C_{k-2}^ix^i}+k\sum_{i=2}^{k-1}{C_{k-2}^{i-1}x^i}=$$
$$=\sum_{i=2}^{k}{C_{k}^ix^i}+k\Big((k-1)x+\sum_{i=2}^{k-1}C_{k-1}^ix^i\Big)=\sum_{i=2}^{k}{C_{k}^ix^i}+k\sum_{i=1}^{k-1}C_{k-1}^ix^i>k\sum_{i=1}^{k-1}C_{k-1}^ix^i.$$

Now, we show that in the inequality (\ref{e10}) the expression in the first factor of the LHS is greater than the expression in the first factor of the RHS, i.e. the validity of the following inequality:
$$\sum_{i=1}^{2k}{C_{2k}^ix^i}\geq 2k\sum_{i=0}^{k}{C_{k}^ix^{i+1}} = 2k\sum_{i=1}^{k+1}{C_{k}^{i-1}x^i}.$$ To do this, we prove that $C_{2k}^i\geq 2k C_{k}^{i-1}$. Indeed, we rewrite the last inequality as

\begin{equation}\label{e11}
(2k-1)!\cdot (k+1-i)! \geq i\cdot k!\cdot(2k-i)!.
\end{equation}
Using the mathematical induction, we show the validity of the inequality (\ref{e11}) for $i\geq3$. It is easy to show the validity of (\ref{e11}) for $i=3$.
Assume that inequality (\ref{e11}) is true for $i=p$.
Let us prove that (\ref{e11}) is also true for $i=p+1$:
$$(2k-1)!\cdot (k-p)! \geq (p+1)\cdot k!\cdot(2k-p-1)!$$
Using the inequality for $i=p$, we get
$$(p+1)\cdot k!\cdot(2k-p-1)!\leq\frac{(p+1)(2k-1)!\cdot (k+1-p)!}{p (2k-p)}.$$
Hence, it suffices to show that
$$\frac{(p+1)(2k-1)!\cdot (k+1-p)!}{p(2k-p)}\leq(2k-1)!\cdot (k-p)!. $$
From the last inequality, after some algebras, we obtain $k(p-1)\geq0.$
So inequality (\ref{e10}) is true.
This means that all coefficients of $x^n$, $n\geq4$, in the polynomial $\mathcal{R}_1(x,k)$ are positive.
It follows that
$$\mathcal{R}(t,k)=(t-1)^4\cdot\mathcal{L}(t),$$
and $\mathcal{L}(t)>0$ for $t>0$. It means that
$$\lambda'(t)=\frac{(t-1)^4(t^k+1)^{k-1}\cdot\mathcal{L}(t)}
{t^2(t^k-1)(t^{k-1}-1)^2\cdot(t^{k-1}+t^{k-2}+\dots+t+1) ^k}.$$
Therefore, the function $\lambda(t)$ decreases for $t<1$, increases for $t>1$ and reaches its minimum for $t=1$ (see Fig. \ref{fig2}):
$$\lambda_{\min}(t)=\lambda(1)=\lambda_{cr}^{(1)}(k)=\frac{1}{k-1}\left(\frac{2 }{k}\right)^{k}.$$

\begin{figure}[h]
 \center{\includegraphics[width=5cm]{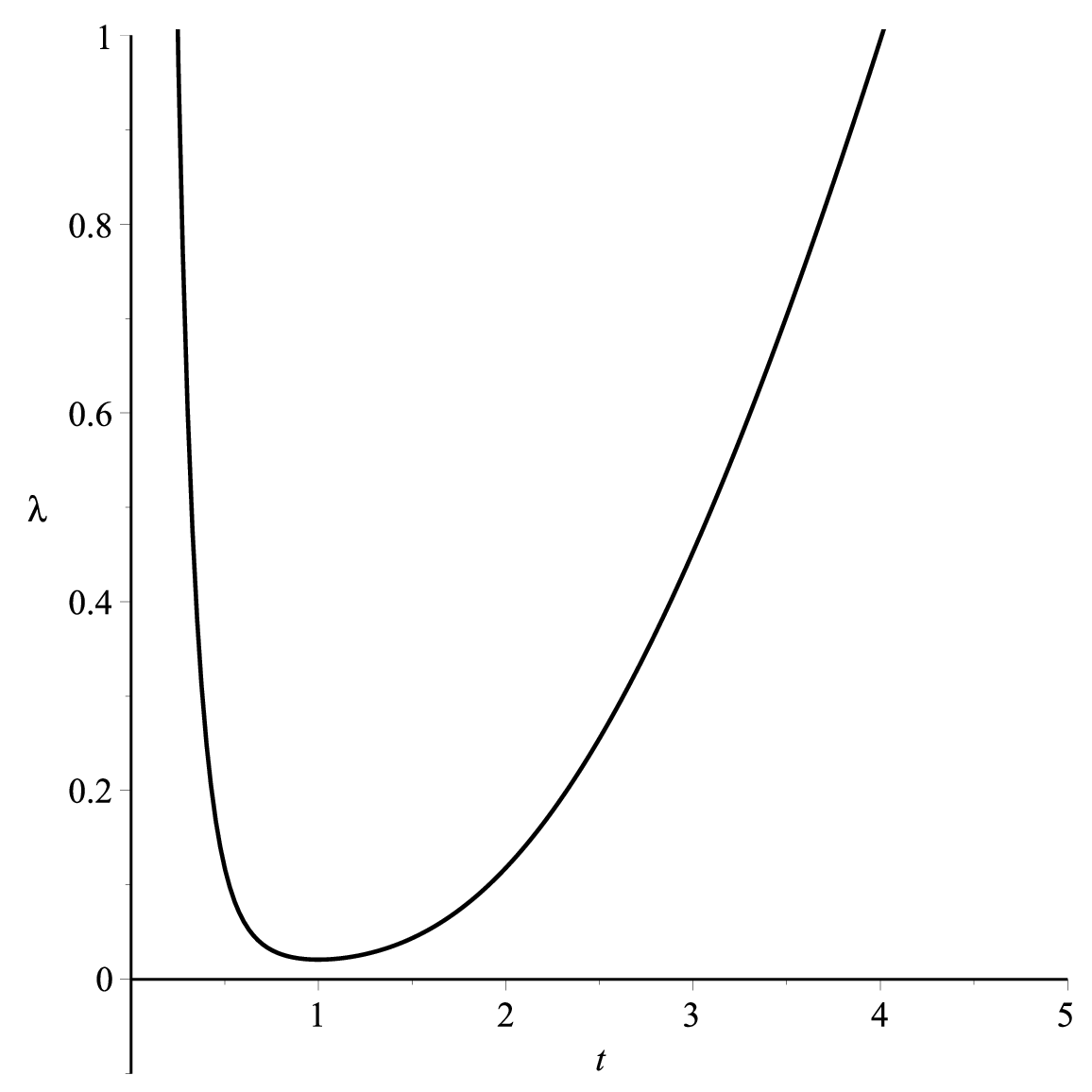}}
\label{fig2}
\caption{Graph of the function $\lambda(t)$ for $k=4$.}
\end{figure}

Hence, each value $\lambda$ corresponds to only one value $t>1$ (or $t<1$) for $\lambda>\lambda_{cr}^{(1)}$, the value $t=1$ for $\lambda=\lambda_{cr}^{(1)}$ and equation (\ref{e5}) has no solutions for $\lambda<\lambda_{cr}^{(1)}$. Due to symmetry, the system of equations (\ref{e5}) for $\lambda\leq\lambda_{cr}^{(1)}$ has a unique solution of the form $(z^*,z^*)$, and for $\lambda>\lambda_{cr}^{(1)}$ has exactly three solutions $(z^*,z^*)$, $(z_{1},z_{2})$ and $(z_{ 2},z_{1})$, $z_{1}\neq z_{2}$, where $z^*$ is the unique positive solution of (\ref{e5}). Denote the TISGMs corresponding to the solutions $(z^*,z^*)$, $(z_{1},z_{2})$ and $(z_{2},z_{1})$ by $\mu_0$, $\mu_1$ and $\mu_2$, respectively..
\end{proof}

Conjecture 1 given in \cite{RKh1} is also true, i.e. the following theorem is true.
\begin{thm}
Let $k\geq2$ and $\lambda_{cr}(k)=\frac{(k+1)^{k}}{(k-1)k^k}$. Then for the HC model in the case of $G=\textit{Hinge}$ for $0<\lambda\leq\lambda_{cr}$ there is exactly one TISGMs $\nu_0$, and for $\lambda>\lambda_{cr}$ there are exactly three TISGMs $\nu_0$, $\nu_1$ and $\nu_2$.
\end{thm}

\begin{proof}
The proof is similar to the proof of Theorem \ref{thm2}.
\end{proof}

\section{Alternative Gibbs measures in the case $G=\textit{Wand}$}

We consider the half-tree. Namely the root $x^0$ has $k$ nearest neighbors. We construct below new solutions of the functional equation (\ref{e4}). Consider the following matrix
$$M=\begin{pmatrix} m & k-m \\ r & k-r \end{pmatrix}$$
where  $0\leq m\leq k$ and $0\leq r\leq k$ are non-negative integers.

This matrix defines the number of times the values $z$ and $t$ occur in the set $S(x)$ for each $z_x\in \{z,t\}$. More precisely, the boundary condition
$z=\{z_x,x\in G_k\}$ with fields taking  values $z=(z_1,z_2)\in R_{+}^{2}$ and $t=(t_1,t_2)\in R_{+}^{2}$ defined by the following steps:

$\bullet$ if at vertex $x$ we have $z_x=z$, then the function $z_y$, which gives real values to each vertex $y\in S(x)$ by the following rule
$$\begin{cases}
z ~\mbox{on} ~m~ \mbox{vertices of} ~ S(x),\\
t ~\mbox{on} ~k-m ~\mbox{remaining vertices,}
\end{cases}$$

$\bullet$ if at vertex $x$ we have $z_x=t$, then the function $z_y$, which gives real values to each vertex $y\in S(x)$ by the following rule
$$\begin{cases}
t ~\mbox{on} ~r~ \mbox{vertices of} ~ S(x),\\
z ~\mbox{on} ~k-r ~\mbox{remaining vertices.}
\end{cases}$$
For an example of such a function see Fig. \ref{fig3}.

\begin{figure}[h] 
\scalebox{1}{
\begin{tikzpicture}[level distance=2.0cm,
level 1/.style={sibling distance=1.8cm},
level 2/.style={sibling distance=2.0cm},
level 3/.style={sibling distance=.2cm}]
\node {$z$} [grow'=up]
    child[sibling distance=2.9cm] {node {$t$}
        child[sibling distance=.3cm] foreach \name in {z,z,z,z,t} { node
        {$\name$} }
        }
        child[sibling distance=2.8cm] {node {$t$}
            child[sibling distance=.3cm] foreach \name in {z,z,z,z,t} { node
        {$\name$} }
        }
        child[sibling distance=2.0cm] {node {$z$}
        child[sibling distance=.3cm] foreach \name in {t,t,z,z,z} { node
        {$\name$} }
    }
    child[sibling distance=2.8cm]{node {$z$}
        child[sibling distance=.3cm] foreach \name in {t,t,z,z,z} { node
        {$\name$} }
    }
    child[sibling distance=2.9cm] {node {$z$}
        child[sibling distance=.3cm] foreach \name in {t,t,z,z,z} { node
        {$\name$} }
    }
;
\end{tikzpicture}
}
\label{fig3}
\caption{In this figure the values of function $z_x$ on the vertices of the Cayley tree
of order 5 are shown in the case $m=3$ and $r=1$.} 
\end{figure}
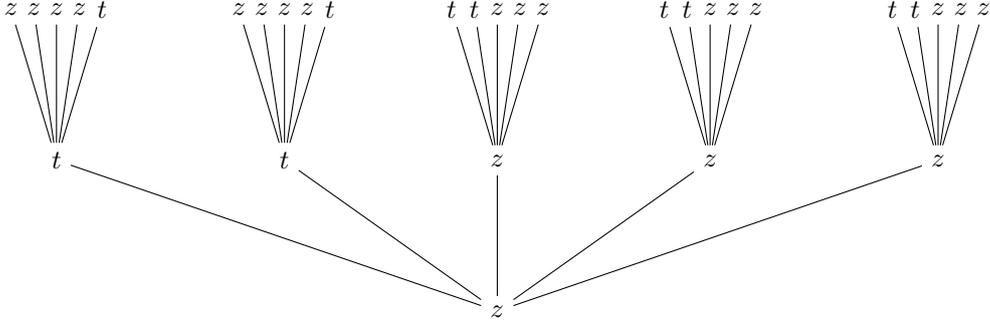

Then the system (\ref{e4}) in the case $G=\textit{wand}$ has the form
\begin{equation} \label{e12}
\begin{cases}
    z_1=\lambda\cdot \left({1+z_1\over z_1+z_2}\right)^m\cdot \left({1+t_1\over t_1+t_2}\right)^{k-m}, \\
    z_2=\lambda\cdot \left({1+z_2\over z_1+z_2}\right)^m\cdot \left({1+t_2\over t_1+t_2}\right)^{k-m}, \\
    t_1=\lambda\cdot \left({1+t_1\over t_1+t_2}\right)^r\cdot \left({1+z_1\over z_1+z_2}\right)^{k-r}, \\
    t_2=\lambda\cdot \left({1+t_2\over t_1+t_2}\right)^r\cdot \left({1+z_2\over z_1+z_2}\right)^{k-r},
\end{cases}
\end{equation}
where $z(z_1,z_2)>0, \ \  t(t_1,t_2)>0,\ \  \lambda>0.$

As was mentioned above, for any boundary condition satisfying the functional equation (\ref{e4}) there exists a unique Gibbs measure. A measure constructed in this way and which is not translation invariant (TI) is called alternative Gibbs measure (AGM) and denoted as $\mu_{m,r}$.

\begin{rk}\label{rm1}
  Note that the solution $z=t$ (i.e. $(z_{1},z_{2})=(t_1,t_2)$) in (\ref{e12}) corresponds to the unique TISGM $\mu_0$ for the \emph{HC}-model (see \cite{Ro}). Therefore, we are interested in solutions of the form $z\neq t.$
\end{rk}

\begin{rk}\label{rm2}
From (\ref{e12}) for $m=r=0$ we obtain a system of equations whose solutions correspond to the $G^{(2)}_k$-periodic splitting Gibbs measures for the \emph{HC}-model, where $G^{(2)}_k$ is a subgroup of words of even length in $G_k$ (see \cite{KhU1}).
\end{rk}

It is obvious that the following sets are invariant with respect to the operator $W:R^4\rightarrow R^4$ defined by RHS of (\ref{e12}):
%
%
%
$$I_1=\{(z_1,z_2,t_1,t_2)\in R^4:z_1=z_2=t_1=t_2\},\  \ I_2=\{(z_1,z_2,t_1,t_2)\in R^4:z_1=t_1,z_2=t_2\},\ $$
$$I_3=\{(z_1,z_2,t_1,t_2)\in R^4:m=r; z_1=t_2,z_2=t_1\},  \ \ I_4=\{(z_1,z_2,t_1,t_2)\in R^4:z_1=z_2, t_1=t_2\},$$
%

\begin{rk}\label{rm3}
In general, the analysis of a system of equations (\ref{e12}) is  a very hard. Therefore, we consider it on invariant sets $I_i, i=1,2,3,4$. Note that the mapping $W$ could have another invariant sets which are different from the sets $I_i, i=1,2,3,4$.
\end{rk}

\begin{rk}\label{rm4} In the case of $I_1$, from the system of equations (\ref{e12}) we obtain the equation (\ref{e6}), which, due to Lemma 1, has a unique positive solution for $k\geq2$ and $\lambda>0$. Moreover, this solution corresponds to the unique TISGM $\mu_{0}$. Thus, the set of values from $I_1$ corresponds to a unique TISGM $\mu_{0}$.
\end{rk}

\begin{rk}\label{rm5}
In the case of $I_2$ from (\ref{e12}) we obtain a system of equations whose solutions correspond to TISGMs is studied in the works  \cite{XR1}, \cite{RKh1} and \cite{Ro}.
\end{rk}

\subsection{Alternative Gibbs measures on $I_3$}

On $I_3$ the system of equations (\ref{e12}) has the form
\begin{equation}\label{e14}
\begin{cases}
 z_1=\lambda\cdot\left({1+z_1\over z_1+z_2}\right)^m\cdot\left({1+z_2\over z_1+z_2}\right)^{k-m},\\
 z_2=\lambda\cdot\left({1+z_2\over z_1+z_2}\right)^m\cdot\left({1+z_1\over z_1+z_2}\right)^{k-m}.
\end{cases}
\end{equation}

The following theorem is true.
\begin{thm}\label{thm4}
Let $k\geq2$ and $\lambda>0$. If $k\geq 2m-1$ then for the HC model in the case of $G=\textit{wand}$ on $I_{3}$ there is no AGM except for TI.
\end{thm}

\begin{proof}
We show that under the conditions of the theorem, the system of equations (\ref{e14}) has a unique solution of the form $z_1=z_2$, which corresponds to TISGM $\mu_{0}$. If the first equation (\ref{e14}) is divided by the second, then
\begin{equation} \label{e15}
z_1(1+z_1)^{k-2m}=z_2(1+z_2)^{k-2m}.
\end{equation}

It is easy to check that the function $f(x)=x(1+x)^{k-2m}$ is increasing for $k\geq2m-1$.
Therefore, if $k\geq2m-1$, then the system of equations (\ref{e14}) has only a solution
of the form $z_1=z_2$, and due to Remark \ref{rm4}, such a solution corresponds to a unique TISGM $\mu_{0}$, i.e. there is no AGM except from TI.
\end{proof}

\begin{thm}\label{thm5}
Let $k\geq2$. For the \emph{HC}-model in the case $G=\textit{wand}$ on $I_{3}$ for some $\lambda>0$ there are AGMs (non TI) if and only if $2k>2m\geq k+2$.
\end{thm}

\begin{proof}
\emph{Necessity.} It follows from Theorem \ref{thm4}.

\emph{Sufficiency.} Let $2m-k=n$  $(n\geq2).$ Then from (\ref{e15}) after simple algebra, we get
$$(z_1-z_2)\cdot\left[z_1z_2\Big(C_n^2+C_n^3(z_1+z_2)+C_n^4(z_1^2+z_1z_2+z_2^2)+\cdots+C_n^n(z_1^{n-2}+\cdots+z_2^{n-2})\Big)-1\right]=0.$$
Hence $z_1=z_2$ or $w(z_1,z_2)=0$, where
$$w(z_1,z_2):=z_1z_2\Big(C_n^2+C_n^3(z_1+z_2)+C_n^4(z_1^2+z_1z_2+z_2^2)+\cdots+C_n^n(z_1^{n-2}+\cdots+z_2^{n-2})\Big)-1=0.$$

If $z_1=z_2$, then we obtain a solution from $I_1$ and, by Remark \ref{rm4}, this solution corresponds to the TISGM $\mu_{0}$.

Let $z_1\neq z_2$. Consider the equation $w(z_1,z_2)=0$ with respect to the variable $z_1$. Then it is clear that $w(0,z_2)=-1<0$ and $w(z_1,z_2)\rightarrow+\infty$ as $z_1\rightarrow+\infty$. Then the equation $w(z_1,z_2)=0$ for variable $z_1$ has at least one positive root. On the other hand, according to Descartes' theorem, the equation $w(z_1,z_2)=0$ for variable $z_1$ has at most one positive root.
Consequently, the equation $w(z_1,z_2)=0$ for variable $z_1$ has exactly one positive root, i.e. there is a solution $(z_1,z_2)$ to the system of equations (\ref{e14}) except for $(z_1,z_1)$ (i.e. $z_1=z_2$) and it corresponds to the AGM. Due to symmetry, the same arguments also apply to the variable $z_2$.

In particular, for $2m-k=2 \ (n=2)$ there exists an AGM (not TI) if $z_1 z_2=1,$ and for $2m-k=3\ (n=3)$ such a measure exists if $z_1^2z_2+z_1z_2^2+3z_1z_2=1.$
\end{proof}

We consider the case $k=4$ and $m=3$, for which the inequalities $2k>2m\geq k+2$ hold.
In this case, the system of equations (\ref{e14}) has the form
\begin{equation} \label{e16}
\begin{cases}
    z_1=\lambda\left(\frac{1+z_1}{z_1+z_2}\right)^3\cdot\frac{1+z_2}{z_1+z_2},\\
    z_2=\lambda\left(\frac{1+z_2}{z_1+z_2}\right)^3\cdot\frac{1+z_1}{z_1+z_2}.
\end{cases}
\end{equation}

\begin{pro}\label{p1}
Let $\lambda_{cr}=1$. Then the system of equations (\ref{e16}) for $\lambda\leq\lambda_{cr}$ has a unique solution $(z_1,z_1)$, and for $\lambda>\lambda_{cr}$ it has exactly three solutions $ (z_1,z_1)$, $(z_{1},z_{2})$, $(z_{2},z_{1})$.
\end{pro}

\begin{proof}
From the system of equations (\ref{e16}) for $n=2$ we have $(z_1-z_2)(z_1z_2-1)=0.$ Hence $z_1=z_2$ or $z_1z_2=1.$ It is known that the solution $z_1 =z_2$ corresponds to the only TISGM.

Let $z_1=\frac1{z_2}$, for $z_1\neq z_2$. Then from (\ref{e16}) we obtain
\begin{equation} \label{e17}
\lambda(z_2)=\frac{(z_2^2+1)^4}{z_2^2(z_2+1)^4}.
\end{equation}
Analyzing the function $\lambda(z_2)$, we have that the function $\lambda(z_2)$ decreases for $z_2<1$, increases for $z_2>1$ and at the $z_2=1$ it reaches its minimum:
$$\lambda_{\min}(z_2)=\lambda(1)=\lambda_{cr}=1.$$
This means that each value of $\lambda$ corresponds to only one value $z_2>1$ (or $z_2< 1$) for
$\lambda >\lambda_{cr}$, $z_2 = 1$ for $\lambda=\lambda_{cr}$ and equation (\ref{e17}) has no solutions for $\lambda<\lambda_{cr} $.
Hence, due to symmetry, the system of equations (\ref{e16}) for $0<\lambda\leq\lambda_{cr}$ has a unique solution
of the form $(z_1; z_1)$, and for $\lambda >\lambda_{cr}$ has exactly three solutions $(z_1; z_1)$, $(z_1; z_2)$ and $(z_2; z_1)$, where $z_1\neq z_2$.
\end{proof}

By virtue of Theorem \ref{thm5} and Proposition \ref{p1}, the following theorem holds
\begin{thm}\label{thm6}
Let $k=4$, $m=3$ and $\lambda_{cr}=1$. Then for the HC model in the case of $G=\textit{wand}$ on $I_{3}$ for $\lambda\leq\lambda_{cr}$ there is a unique AGM, which coincides with the only TI $\mu_0$, and for $\lambda>\lambda_{cr}$ there are exactly three AGMs $\mu_0$, $\mu_1$ and $\mu_2$, where $\mu_1$ and $\mu_2$ are AGMs (not TI).
\end{thm}

\subsection{Alternative Gibbs measures on $I_4$}

We denote $z_1=z_2=z, \ t_1=t_2=t$. Then the system (\ref{e12}) on $I_{4}$ has the following form:
\begin{equation} \label{e18}
\begin{cases}
z=\lambda\left(\frac{1+z}{2z}\right)^{m}\cdot\left(\frac{1+t}{2t}\right)^{k-m},\\
t=\lambda\left(\frac{1+t}{2t}\right)^{r}\cdot\left(\frac{1+z}{2z}\right)^{k-r}.
\end{cases}
\end{equation}

\begin{thm}\label{thm7}
Let $k\geq2$ and $m+r\geq k-1$. Then for the HC model in the case of $G=\textit{wand}$ on $I_{4}$ there is no AGM except for TI.
\end{thm}

\begin{proof}
Let us show that under the conditions of the theorem, the system of equations (\ref{e18}) has a solution only of the form $z=t$, which corresponds to the TISGM $\mu_0$. If the first equation (\ref{e18}) is divided by the second, we get
\begin{equation} \label{e19}
z\left(\frac{2z}{1+z}\right)^{m+r-k}=t\left(\frac{2t}{1+t}\right)^{m+r-k}.
\end{equation}

It is easy to see that if $m+r=k-1$ and $m+r=k$ then the equation (\ref{e19}) has a solution of the form $z=t$.

Let $m+r>k.$ Then from (\ref{e19}) we get
$$\frac{z}{t}=\left(\frac{t}{z}\cdot\frac{1+z}{1+t}\right)^{n},$$
where $m+r-k=n,$ $(n\geq-1)$. Hence, denoting $\sqrt[n]{z}=x$, $\sqrt[n]{t}=y$, after some algebras, we obtain
$$\left(x-y\right) \left(x^{n}+\cdots+y^{n}+x^{n}y^{n}\right)=0,$$
which solution has the form $x=y$, and this solution corresponds to the TISGM $\mu_0$.
\end{proof}

\begin{thm}\label{thm8}
Let $k\geq2$. For the \emph{HC}-model in the case $G=\textit{wand}$ on $I_{4}$ for some $\lambda>0$ there are AGMs (non TI) if and only if $m+r\leq k-2$.
\end{thm}
\begin{proof}
The proof is similar to the proof of Theorem \ref{thm5}.
\end{proof}

In particular, for $m+r=k-2 \ (n=2)$ there is an AGM (not TI) if $z t=1,$ and for $m+r=k-3\ (n=3)$ such a measure exists if $\sqrt[3]{t^2z^2}=\sqrt[3]{z}+\sqrt[3]{t}.$

Let $m+r\leq k-2$.

\textbf{Case $k=3, m=1, r=0$} (resp. $m=0$ and $r=1$). In this case from the system of equations ( \ref{e18}) we get
\begin{equation} \label{e21}
\begin{cases}
    z=\lambda\frac{1+z}{2z}\left(\frac{1+t}{2t}\right)^{2},\\
    t=\lambda\left(\frac{1+z}{2z}\right)^{3}.
\end{cases}
\end{equation}

\begin{pro}\label{p2}
Let $\lambda_{cr}^{(2)}=\frac{32}{27}$. Then the system of equations (\ref{e21}) for $\lambda>\lambda_{cr}^{(2)}$ has a unique solution $(z,z)$, for $\lambda=\lambda_{cr}^{ (2)}$ has two solutions $(z,z)$, $(2;\frac{1}{2})$ and for $0<\lambda<\lambda_{cr}^{(2)}$ has three solutions $(z,z)$, $(z_{1},t_{1})$, $(z_{2},t_{2})$.
\end{pro}

\begin{proof} From the system of equations (\ref{e21}) we get $(z-t)(z t-1)=0.$ Hence $z=t$ or $zt=1.$ By virtue of Remark \ref{rm4}, the solution $z=t$ corresponds to the unique TISGM $\mu_0$.

Let $zt=1$, $z\neq t.$ Then from the first equation of the system (\ref{e21}) we obtain
\begin{equation} \label{e22}
 z^3+\left(3-\frac{8}{\lambda}\right)z^{2}+3z+1=0.
\end{equation}

Let us find positive solutions of equation (\ref{e22}) using Cardano's formula.
Let $z=y+\frac{8}{3\lambda}-1.$ Then
\begin{equation} \label{e23}
y^3-\frac{16\left(3\lambda-4\right)}{3\lambda^{2}} y-\frac{8(27\lambda^{2}-144\lambda+128)}{27\lambda^{3}}=0.
\end{equation}
For this equation we calculate
$$D=\frac{16(27\lambda-32)}{27\lambda^{3}}.$$

It is known that for $\lambda>\frac{32}{27}$ ($D>0$) the equation (\ref{e23}) has one negative root, for $\lambda=\frac{32}{27} $ ($D=0$) multiple positive root   $y=\frac{3}{4}$, (i.e. $t_0=\frac12$, $z_0=2$) and for $\lambda<\frac{32}{27}$ ($D<0$) has three real roots.
Hence, $f(y)=0$ has three real roots if $\lambda<\frac{32}{27}$. Let these solutions be $z_1, z_2, z_3$. By the Vieta's formulas
\begin{equation} \label{e24}
z_1+z_2+z_3=\frac{8}{\lambda}-3, \qquad  z_1z_2+z_1z_3+z_2z_3=3, \qquad  z_1z_2z_3=-1.
\end{equation}
From the third equality (\ref{e24}) we deduce that the equation have either two positive solutions or does not have any positive solutions. If all three factors are negative, then their sum also must be negative. In this case,  from the first equality (\ref{e24}) we obtain $\lambda>\frac{8}{3}$. On the other hand, the roots $z_1, z_2, z_3$ exist for $\lambda<\frac{32}{27}$. This means that the equation cannot have three negative solutions. It means that one of the roots is negative, and the other two are positive.
\end{proof}
Due to Theorem \ref{thm8} and Proposition \ref{p2}, we obtain the following theorem.
\begin{thm}
Let $k=3$, $r+m\leq 1$ and $\lambda_{cr}^{(2)}=\frac{32}{27}$. Then for the HC model in the case of $G=\textit{wand}$ on the $I_{4}$ for $\lambda>\lambda_{cr}^{(2)}$ there is a unique AGM which coincides with the only TISGM $\mu_0$, for $\lambda=\lambda_{cr}^{(2)}$ there are two AGMs $\mu_0$ and $\mu^{'}$, where $\mu^{'}$ is AGM (not TI), and for $0<\lambda<\lambda_{cr}^{(2)}$ there are exactly three AGMs $\mu_0$, $\mu_1$ and $\mu_2$, where $\mu_1$ and $ \mu_2$ are AGMs (not TI).
\end{thm}
\textbf{Case $k=4, m=1, r=0$} (resp. $m=0$, $r=1$). In this case from (\ref{e18}) we get

\begin{equation} \label{e25}
\begin{cases}
    z=\lambda\frac{1+z}{2z}\left(\frac{1+t}{2t}\right)^{3},\\
    t=\lambda\left(\frac{1+z}{2z}\right)^{4}.
\end{cases}
\end{equation}

\begin{pro}\label{p3}
   Let $\lambda_{cr}^{(3)}\approx6.913562404$. Then the system of equations (\ref{e25}) for $\lambda>\lambda_{cr}^{(3)}$ has a unique solution $(z,z)$, for $\lambda=\lambda_{cr}^{ (3)}$ has two solutions $(z,z)$, $(z^{'},t^{'})$ and for $0<\lambda<\lambda_{cr}^{(3)}$ has three solutions $(z,z)$, $(z_{1},t_{1})$, $(z_{2},t_{2})$.
\end{pro}

\begin{proof}
Note that, by Remark \ref{rm4}, the solution $z=t$ corresponds to a unique TISGM. Therefore, it it sufficient to consider $z\neq t$.

In the system of equations (\ref{e25}), substituting the expression of $t$ from the second equation into the first, we obtain
$$ z=\lambda\frac{1+z}{2z}\left(\frac{16z^4+\lambda(1+z)^4}{2\lambda(1+z)^{4}} \right)^{3}.$$

Due to the complexity of solving the last equation with respect to $z$, we consider it as an equation with respect to the parameter $\lambda$. Then
$$\lambda^{3}\left(1+z\right)^{12}-16\lambda^{2}z^2(1+z)^8\Big(1+3z+z^3\Big)+768\lambda z^{8}(1+z)^{4}+4096z^{12}=0,$$
which solutions have the form:
$$\lambda_{1}(z)=\frac{16z^{5}}{(1+z)^{4}}, \qquad \lambda_{2}(z)=\frac{\left(24z+ 8-8\sqrt{4z^{3}+9z^{2}+6z+1}\right)z^{2}}{(1+z)^{4}},$$
\begin{equation} \label{e26}
\lambda_{3}(z)=\frac{\left(24z+8+8\sqrt{4z^{3}+9z^{2}+6z+1}\right)z^{2}}{( 1+z)^{4}}.
\end{equation}
Note that $\lambda_{1}(z)$ is obtained from (\ref{e25}) for $z=t$ and this solution corresponds to the TISGM $\mu_0$. In addition, it is easy to see that $\lambda_{2}(z)<0$ and $\lambda_{3}(z)>0$ for $z>0$.

We calculate the derivative of the function $\lambda_{3}(z)$. Then from the equation $\lambda_{3}^{'}(z)=0$ we get
\begin{equation} \label{e27}
\Big(3z^2-7z-2\Big)\sqrt{4z^{3}+9z^{2}+6z+1}=-\Big(2z^2-9z-2\Big)(1+z)^{2}.
\end{equation}
Both sides of the (\ref{e27}) can be squared if
\begin{equation} \label{e30}
\frac{7+\sqrt{73}}{6}\leq z\leq\frac{9+\sqrt{97}}{4}.
\end{equation}
Given the condition (\ref{e30}), let us square both sides of the equation (\ref{e27}). As a result, after some algebras we have
$$z^3(z^3-16z^2+41z+10)(1+z)^2=0.$$
Note that the roots of the equation (\ref{e27}), satisfying the condition (\ref{e30}), are among the solutions of the last equation. We find them using the Cardano formula:
$$z_1=\frac23\sqrt{133}\cdot\cos{\left(\frac{1}{3}\arctan{\frac{18}{1009}\sqrt{4119}}\right)}+\frac{16}{3}\approx12.71311048,$$
$$z_2=\frac23\sqrt{133}\cdot\cos{\left(\frac{1}{3}\arctan{\frac{18}{1009}\sqrt{4119}}+\frac{2\pi}{3}\right)}+\frac{16}{3}\approx-0.224040245,$$
$$z_3=\frac23\sqrt{133}\cdot\cos{\left(\frac{1}{3}\arctan{\frac{18}{1009}\sqrt{4119}}+\frac{4\pi}{3}\right)}+\frac{16}{3}\approx3.510929776.$$

It is clear that $\frac{7+\sqrt{73}}{6}\leq z_3\leq\frac{9+\sqrt{97}}{4}$ and the function $\lambda_{3}(z)$ at point $z_3$ reaches its maximum:
$$\lambda_{cr}^{(3)}=\lambda_{3}(z_3)\approx6.913562404.$$

Thus, each value of $\lambda_{3}$ in (\ref{e26}) corresponds to two values of $z$ for $0<\lambda<\lambda^{(3)}_{cr}$, one value of $z $ for $\lambda=\lambda^{(3)}_{cr}$ and equation (\ref{e26}) has no solution for $\lambda>\lambda^{(3)}_{cr}$ (see Fig. 4.).
\end{proof}

 \begin{figure}[h]
\center{\includegraphics[width=5cm]{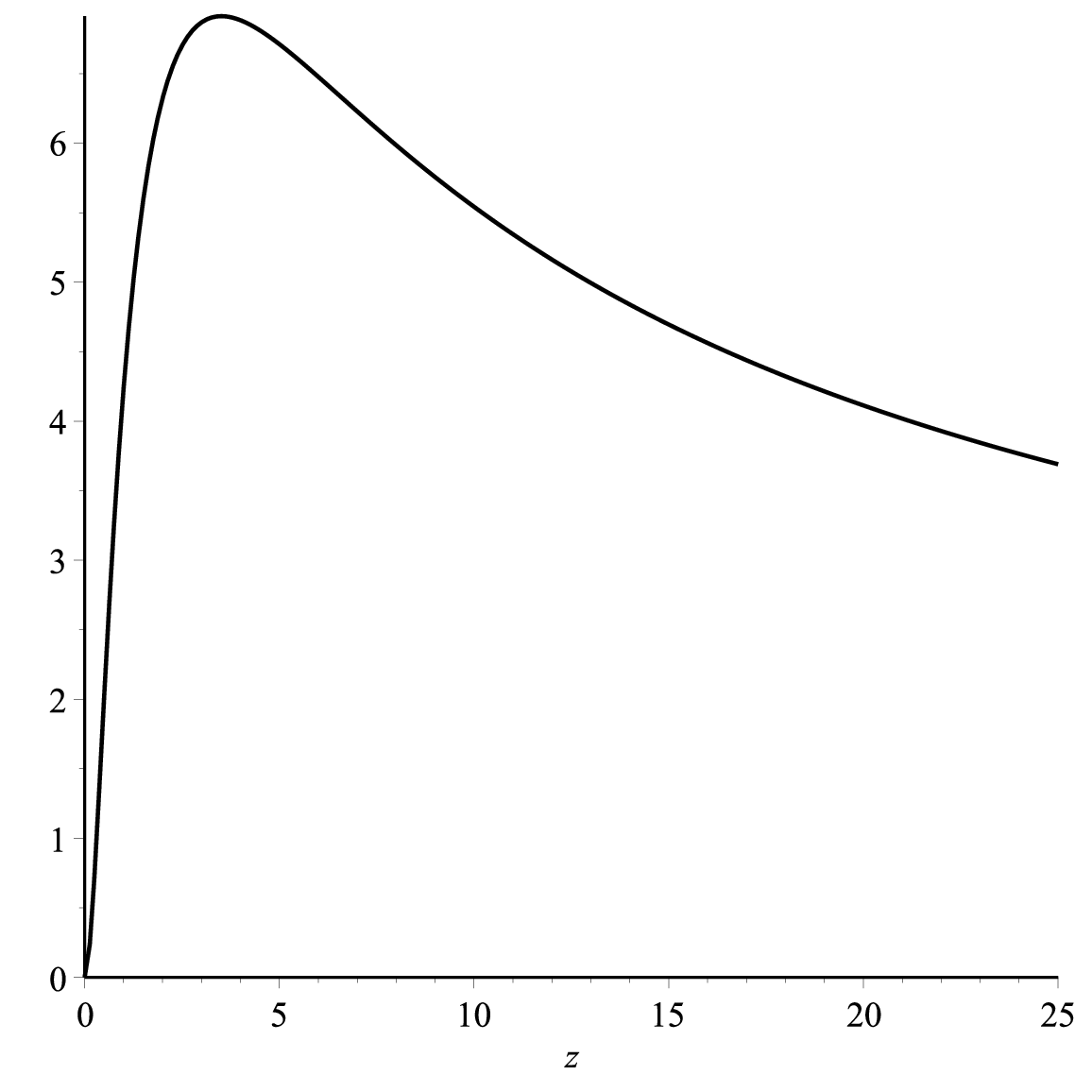}}
\label{fig4}
\caption{Graph of the function $\lambda_3(z)$. }
\end{figure}

\textbf{Case $k=4, m=1, r=1$.} In this case, the system of equations (\ref{e18}) has the form
\begin{equation} \label{e31}
\begin{cases}
    z=\lambda\left(\frac{1+z}{2z}\right)\left(\frac{1+t}{2t}\right)^{3},\\
    t=\lambda\left(\frac{1+t}{2t}\right)\left(\frac{1+z}{2z}\right)^{3}.
\end{cases}
\end{equation}

\begin{pro}\label{p4}
Let $\lambda_{cr}^{(4)}=1$. Then the system of equations (\ref{e31}) for $\lambda\geq\lambda_{cr}^{(4)}$ has a unique solution $(z,z)$, and for $0<\lambda<\lambda_{cr }^{(4)}$ has three solutions $ (z,z),\ (z_{1},z_{2}), \ (z_{2},z_{1}).$
\end{pro}

\begin{proof}
Dividing the first equation in (\ref{e31}) by the second, we get
$$(z-t)(z\cdot t-1)=0.$$
Hence $z=t$ or $z\cdot t=1.$ Note that the solution $(z,z)$ corresponds to a unique TISGM.

Let $z\neq t$, i.e. $t=\frac{1}{z}$. Then from the second equation of the system of equations (\ref{e31}) we have
\begin{equation} \label{e32}
\lambda(1+z)^{4}-16z^{2}=0,
\end{equation}
which solutions have the form
$$z_1=\frac{2-\sqrt{\lambda}-2\sqrt{1-\sqrt{\lambda}}}{\sqrt{\lambda}}, \qquad z_2=\frac{2-\sqrt{ \lambda}+2\sqrt{1-\sqrt{\lambda}}}{\sqrt{\lambda}}.$$
It is clear that $0<\lambda\leq1=\lambda_{cr}^{(4)}$.
For $\lambda=1$ we have a unique solution $z_1=z_2$, which corresponds to a unique TISGM.

For $0<\lambda<\lambda_{cr}^{(4)}$, due to the equality $zt=1$ we have
$$t_1=\frac{1}{z_1},\ t_2=\frac{1}{z_2}.$$

In addition, it is obvious that $z_1\cdot z_2=1$. Hence, for $0<\lambda<\lambda_{cr}^{(4)}$, the system of equations (\ref{e32}) has two solutions of the form $(z_{1},z_{2})$ and $(z_{ 2},z_{1}).$
\end{proof}

\textbf{Case $k=4, m=2, r=0$.} In this case, the system of equations (\ref{e18}) has the form
\begin{equation} \label{e33}
\begin{cases}
    z=\lambda\left(\frac{1+z}{2z}\right)^{2}\left(\frac{1+t}{2t}\right)^{2},\\
    t=\lambda\left(\frac{1+z}{2z}\right)^{4}.
\end{cases}
\end{equation}

\begin{pro}\label{p5}
Let $\lambda_{cr}^{(5)}=\frac{27}{16}$. Then the system of equations (\ref{e33}) for $\lambda>\lambda_{cr}^{(5)}$ has a unique solution $(z,z)$, for $\lambda=\lambda_{cr}^{ (5)}$ has two solutions $(z,z)$, $(3,\frac{1}{3})$ and for $0<\lambda<\lambda_{cr}^{(5)}$ has three solutions $(z,z)$, $(z_{1},t_{1})$, $(z_{2},t_{2})$.
\end{pro}

\begin{proof}
Similar to the previous case, from (\ref{e33}) we obtain $z=t$ or $z\cdot t=1$.
The case $z=t$ has already been considered.

We consider the case $z\neq t$, i.e. $t=\frac{1}{z}$. Then from the second equation (\ref{e33}), denoting $\sqrt[4]{z}=x$ and $\frac{1}{\sqrt[4]{\lambda}}=a$, we have
\begin{equation} \label{e34}
x^{4}-2ax^{3}+1=0.
\end{equation}

ÀAnalyzing the function $f(x)=x^{4}-2ax^{3}+1$, we obtain that the function $f(x)$ decreases for $x<x_{0}$, and increases for $x>x_{0}$, where $x_0=x_{\min}=\frac{3a}{2}$ and $f_{\min}=f(x_{0})=-\frac{27}{16}a^{4}+1.$ Hence, equation (\ref{e34}) has no solution if $f(x_{0})>0$ (i.e. $\lambda>\frac{27}{16}=\lambda_{cr}^{(5)}$), has one positive solution $x=\sqrt[4]{3}$ if $f(x_{0})=0$ (i.e. $\lambda=\lambda_{cr}^{(5)}$) and has two positive solutions if $f(x_{0})<0$ (i.e. $0<\lambda<\lambda_{cr} ^{(5)}$).
\end{proof}

Using Propositons \ref{p3}, \ref{p4} and \ref{p5}, we get
\begin{thm}
Let $k=4$. Then for the HC model in the case of $G=\textit{wand}$ on the $I_{4}$ the following statements are true:

1. Let $m=1$, $r=0$ or $m=0$, $r=1$. Then there exists $\lambda^{(3)}_{cr}$ $(\approx6.913562)$ such that for $\lambda>\lambda_{cr}^{(3)}$ there is a unique AGM which coincides with TISGM $\mu_0$, for $\lambda=\lambda_{cr}^{(3)}$ there are two AGMs $\mu_0$, $\mu^{'}$, where $\mu^{'}$ is AGM(not TI), and for $0<\lambda<\lambda_{cr}^{(3)}$ there are exactly three AGMs $\mu_0$, $\mu^{'}_1$ and $\mu^{'}_2$, where $\mu^{'}_1$ and $\mu^{'}_2$ are AGM (not TI).

2. Let $m=1$, $r=1$ and $\lambda^{(4)}_{cr}=1$. Then for $\lambda\geq\lambda^{(4)}_{cr}$ there is one AGM which coincides with TISGM $\mu_0$, and for $0<\lambda<\lambda^{(4)}_{ cr}$ there are exactly three AGMs $\mu_0$, $\ddot{\mu}_1$ and $\ddot{\mu}_2$, where $\ddot{\mu}_1$ and $\ddot{\mu}_2$ are AGM (not TI).

3. Let $m=2$, $r=0$ or $m=0$, $r=2$ and $\lambda^{(5)}_{cr}=\frac{27}{16}$. Then for $\lambda>\lambda^{(5)}_{cr}$ there is a unique AGM which coincides with the only TISGM $\mu_0$, for $\lambda=\lambda_{cr}^{(5)}$ there are two AGMs $\mu_0$, $\acute{\mu}$, where $\acute{\mu}$ is AGM (not TI), and for $\lambda<\lambda^{(5) }_{cr}$ there are exactly three AGMs $\mu_0$, $\acute{\mu}_1$ and $\acute{\mu}_2$, where $\acute{\mu}_1 $ and $\acute{\mu}_2$ are AGM (not TI).
\end{thm}

\section{Relation of the Alternative Gibbs measures to known ones}\

\textit{Translation invariant measures.} (see \cite{Ro}, \cite{XR1}, \cite{RKh1}) Such measures correspond to $z_x\equiv z$, i.e. constant functions. These measures are particular cases of our measures mentioned which can be obtained for $m=k$, i.e. $k-m=0$.
In this case the condition (\ref{e4}) reads
\begin{equation}\label{e35}
\begin{cases}
    z_{1}=\lambda\left(\frac{1+z_{1}}{z_{1}+z_{2}}\right)^k, \\
    z_{2}=\lambda\left(\frac{1+z_{2}}{z_{1}+z_{2}}\right)^k.
\end{cases}
\end{equation}
The equation (\ref{e35}) was studied in Section \ref{sec2}.

\textit{Periodic Gibbs measures.} (see \cite{KhU1}) Let $G_k$ be a free product of $k+1$ cyclic groups of the second order with generators $a_1,a_2,...,a_{k+1},$
respectively.

It is known that there exists a one-to-one correspondence between the set of vertices $V$ of the Cayley tree $\Im^k$ and the group $G_k$.

Let $G^{(2)}_k=\{x\in G_k : \mbox{the length of word} ~x~ \mbox{is even} \}.$
Note that $G^{(2)}_k$ is the set of even vertices (i.e. with even distance to the root). Consider the boundary condition $z$ and $t$:
$$z_x=
\begin{cases} z, ~ \mbox{if}~ x\in G^{(2)}_k,\\
 t, ~ \mbox{if}~ x\in G_k\setminus G^{(2)}_k.
\end{cases}$$
and denote by $\mu_{1}$, $\mu_{2}$ the corresponding Gibbs measures.
The $\widehat{G}$- periodic solutions of equation (\ref{e4}) are either translation-invariant ($G_k$- periodic) or $G^{(2)}_k$
-periodic, they are solutions to
$$
\begin{cases}
    z_1=\lambda\cdot \left({1+t_1\over t_1+t_2}\right)^{k}, \\
    z_2=\lambda\cdot \left({1+t_2\over t_1+t_2}\right)^{k}, \\
    t_1=\lambda\cdot \left({1+z_1\over z_1+z_2}\right)^{k}, \\
    t_2=\lambda\cdot \left({1+z_2\over z_1+z_2}\right)^{k}.
\end{cases}
$$
We note that these measures are particular cases of measures of AGM which
can be obtained for $m=r=0$ (see Fig. 5., for $k=4$).

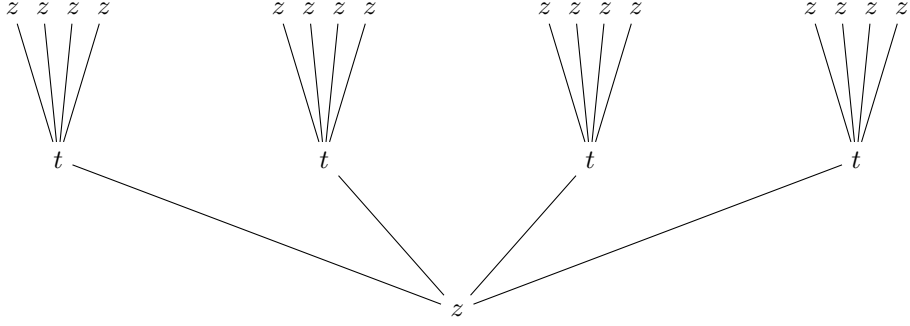
\begin{figure}
\scalebox{1}{
\begin{tikzpicture}[level distance=2.0cm,
level 1/.style={sibling distance=1.8cm},
level 2/.style={sibling distance=2.0cm},
level 3/.style={sibling distance=.2cm}]
\node {$z$} [grow'=up]
    child[sibling distance=3.5cm] {node {$t$}
        child[sibling distance=.4cm] foreach \name in {z,z,z,z} { node
        {$\name$} }
    }
    child[sibling distance=3.5cm] {node {$t$}
        child[sibling distance=.4cm] foreach \name in {z,z,z,z} { node
        {$\name$} }
    }
    child[sibling distance=3.5cm] {node {$t$}
        child[sibling distance=.4cm] foreach \name in {z,z,z,z} { node
        {$\name$} }
    }
    child[sibling distance=3.5cm] {node {$t$}
        child[sibling distance=.4cm] foreach \name in {z,z,z,z} { node
        {$\name$} }
    }
;
\end{tikzpicture}
}
\label{fig5}
\caption {In this figure the values of function $z_x$ on the vertices of the Cayley tree
of order 4 are shown.}
\end{figure}

\textit{Weakly periodic Gibbs measures.} Following \cite{8} we recall the notion of weakly periodic Gibbs measures. Let
$G_k/\widehat{G}_k=\{H_1, ..., H_r\}$ be a factor group, where $\widehat{G}_k$ is a normal subgroup of index $r>1$.

For any $x\in G_k $, the set $\{y\in G_k: \langle
x,y\rangle\}\setminus S(x)$ contains a unique element denoted by $x_{\downarrow}$ (see \cite{8}, \cite{5}).

\begin{defn}
A collection of quantities $z=\{z_x,x\in G_k\}$ is called $\widehat{G}_k$-weakly periodic if $z_x=z_{ij}$ for any $x\in H_i$, $x_\downarrow \in H_j$ for any $x \in G_k$.
\end{defn}

We recall results known for the cases of index two. Note that any such subgroup
has the form
$$H_A=\Big\{x\in G_k:\sum_{i\in A}{w_x(a_i)}~\mbox{is even}\Big\}$$
where ${\emptyset}\neq A\subseteq N_k=\{1,2,...,k+1\}$, and $w_x(a_i)$ is the number of $a_i$ in a word $x\in G_k.$
We consider $A\neq N_k$: when $A=N_k$ weakly periodicity coincides with standard periodicity. Let $G_k/H_A=\{H_0,H_1\}$ be the factor group, where $H_0=H_A$,
$H_1=G_k\setminus H_A$. Then, in view of (\ref{e4}), the $H_A$ - weakly periodic b.c. has the form

$$z_x=
\begin{cases} z, \ x\in H_A, \ x_\downarrow \in H_A, \\
t, \ x\in H_A, \ x_\downarrow \in G_k\setminus H_A,\\
q, \ x\in G_k\setminus H_A, \ x_\downarrow \in H_A, \\
p, \ x\in G_k\setminus H_A, \ x_\downarrow \in G_k\setminus H_A.
\end{cases}$$
which they satisfy the following equations:

\begin{equation}\label{e36}
\begin{cases}
z_1=\lambda\Big(\frac{1+z_1}{z_1+z_2}\Big)^{k-i}\cdot\Big(\frac{1+q_1}{q_1+q_2}\Big)^{i}, \\
z_2=\lambda\Big(\frac{1+z_2}{z_1+z_2}\Big)^{k-i}\cdot\Big(\frac{1+q_2}{q_1+q_2}\Big)^{i}, \\
t_1=\lambda\Big(\frac{1+z_1}{z_1+z_2}\Big)^{k+1-i}\cdot\Big(\frac{1+q_1}{q_1+q_2}\Big)^{i-1}, \\
t_2=\lambda\Big(\frac{1+z_2}{z_1+z_2}\Big)^{k+1-i}\cdot\Big(\frac{1+q_2}{q_1+q_2}\Big)^{i-1},  \\
q_1=\lambda\Big(\frac{1+p_1}{p_1+p_2}\Big)^{k+1-i}\cdot\Big(\frac{1+t_1}{t_1+t_2}\Big)^{i-1},  \\
q_2=\lambda\Big(\frac{1+p_2}{p_1+p_2}\Big)^{k+1-i}\cdot\Big(\frac{1+t_2}{t_1+t_2}\Big)^{i-1},  \\
p_1=\lambda\Big(\frac{1+p_1}{p_1+p_2}\Big)^{k-i}\cdot\Big(\frac{1+t_1}{t_1+t_2}\Big)^{i}, \\
p_2=\lambda\Big(\frac{1+p_2}{p_1+p_2}\Big)^{k-i}\cdot\Big(\frac{1+t_2}{t_1+t_2}\Big)^{i}.
\end{cases}
\end{equation}

It is obvious that the following sets are invariant with respect to the operator $W:R^8\rightarrow R^8$ defined by RHS of (\ref{e36}):
$$I_1=\Big\{z_1=z_2=t_1=t_2=q_1=q_2=p_1=p_2\Big\}, \quad I_2=\Big\{ z_1=t_1=q_1=p_1; \ z_2=t_2=q_2=p_2\Big\}$$
$$I_3=\Big\{z_1=p_1;\ z_2=p_2;\ q_1=t_1;\ q_2=t_2\Big\}, \qquad I_4=\Big\{ z_1=z_2; \ t_1=t_2; \ q_1=q_2; \ p_1=p_2\Big\}$$
It is obvious to see that

$\bullet$ measures corresponding to solutions on $I_1$ and $I_2$ are translation invariant

$\bullet$ measures corresponding to solutions on $I_3$ are weakly periodic, which coincide with the measures given for $m=k-i$, $k-m=i$, $r=i-1$, $k-r=k-i+1$.

\section{data availability statement}
Not applicable

\section{conflicts of interest}
The author declares that they have no conflict of interest.

\end{document}